\documentclass[%
 reprint,
superscriptaddress,
nofootinbib,
 amsmath,amssymb,
 aps,
 pre,
]{revtex4-2}

\usepackage{amsmath,amssymb,amsthm,amsfonts,mathtools}
\usepackage{blindtext}
\usepackage{graphicx}
\usepackage{dcolumn}
\usepackage{bm}
\usepackage{xcolor}
\newtheorem{theorem}{Theorem}

\begin{document}

\preprint{APS/123-QED}

\title{On the Inertial Rotational Brownian Motion of Arbitrarily Shaped Particles}

\author{Amitesh S. Jayaraman}
\email{amiteshs@stanford.edu}
 \affiliation{Department of Mechanical Engineering, Stanford University, CA 94305, USA}
\author{Jikai Ye}%
\email{jikai.ye@nus.edu.sg}
\affiliation{%
Department of Mechanical Engineering, National University of Singapore, 117575, Singapore
}%
\author{Gregory S. Chirikjian}%
\email{mpegre@nus.edu.sg}
\affiliation{%
Department of Mechanical Engineering, National University of Singapore, 117575, Singapore
}%

\date{\today}

\begin{abstract}
This article reports the modeling of inertial rotational Brownian motion as an Ornstein-Uhlenbeck process evolving on the cotangent bundle of the rotation group, $SO(3)$. The benefit of this approach and the use of a different parameterization of rotations allows the handling of particles with arbitrary shapes, without requiring any simplifying assumptions on the shape or the structure of the viscosity tensors. The resultant Fokker-Planck equation for the joint orientation and angular momentum probability distribution can be solved approximately using an `ansatz' Gaussian distribution in exponential coordinates.

\end{abstract}


\maketitle


\section{\label{sec:Introduction}Introduction}

The study of rotational Brownian motion has more than 100 years of history, commencing with Debye's pioneering work \cite{Debye1913} in 1913. Rotational Brownian motion, as opposed to translational Brownian motion, is the phenomenon observed when a three-dimensional body is subject to random torques---giving rise to a random trajectory on the group of rigid body rotations, $SO(3)$. Despite the fact that research in this area commenced more than a century ago, rotational Brownian motion continues to be relevant in a variety of fields beyond soft matter and colloidal science, such as attitude estimation \cite{solo2010attitude} and even exotic applications in astrophysics \cite{merritt2002}. Debye had considered non-inertial rotational Brownian motion and Perrin \cite{Perrin1928} furthered this analysis by casting non-inertial rotational Brownian motion as a stochastic process on the rotation group $SO(3)$.

Obtaining the joint orientation and angular velocity/momentum probability density for the inertial theory of rotational Brownian motion is a more challenging problem for arbitrary body geometries due to the presence of the nonlinear cross-product term in Euler's equations for rigid body motions. Many researchers have previously considered various approximations of the problem. Hubbard \cite{hubbard1972rotational} for instance analyzed the rotational Brownian motion of spherical particles as dynamics evolving on the state space, $SO(3)\times so(3)$. However, his analysis is approximate as well and he performed an iterative solution of the resulting Fokker-Planck equations rather than a direct exact solution. Fixman \cite{fixman1969angular} considered the dynamics over phase space, $SO(3) \times so^{*}(3)$, but his analysis is restricted to symmetric top-shaped particles. Finally, Coffey \cite{coffey2002langevin} considered Brownian motion in the phase space, i.e., $SO(3)\times so^{*}(3)$ for needle-shaped particles but does not solve the Fokker-Planck equations directly. 

There has also been some research in the context of obtaining marginalized versions of the joint probability densities, \textit{i}.\textit{e}., Hubbard \cite{hubbard1977angular} has also considered the probability distribution over angular velocity alone for arbitrarily shaped particles---but the results are most accurate for small deviations from sphericity. Steele on the other hand \cite{steele1963molecular} considers the probability distribution over orientation alone, but focuses primarily on a spherical particle. Ivanov \cite{ivanov1964theory} considers the probability distribution over orientation for arbitrarily shaped particles, with the Brownian motion process modeled as a random walk. 
Finally, Ford et. al., \cite{ford1979rotational} derive orientational and angular velocity autocorrelation functions for asymmetric tops but assume that the viscous tensor and the inertia tensor are simultaneously diagonalizable; moreover, their expressions are derived by a perturbative expansion and are thus approximate. 
Moreover, they do not attempt to directly solve for the probability density function.

This paper considers the inertial rotational Brownian motion problem for an arbitrarily shaped three-dimensional rigid body, with no constraints on the structure of the viscosity tensor or moment of inertia. 
We reformulate the inertial rotational Brownian motion as a stochastic process on the cotangent bundle group of $SO(3)$.
A small-time closed-form expression is derived for the probability density function (PDF) over orientation and angular momentum for the particle. 
The parameters of the PDF are determined by an integration, which has an analytic expression in some cases and can be solved numerically in general.

\section{\label{sec:GoverningEquations} Governing Equations}

\subsection{Equations of Motion}
Consider a three-dimensional rigid body with the moment of inertia tensor $I$ (in a body frame of reference). It is suspended in a fluid medium and at any point of time, its orientation is encoded by a three-dimensional rotation matrix, $R\in SO(3)$. The body frame angular velocity of the body is $\bm{\omega}_R\in\mathbb{R}^3$ and defined such that $\bm{\omega}_R = (R^T\dot{R})^\vee$ where $\vee$ is a bijection converting elements from the Lie algebra of $SO(3)$ to $\mathbb{R}^3$ (in other words, vectorizes $3\times 3$ skew-symmetric matrices). The angular momentum of the body is then $\bm{\ell} = I\bm{\omega}_R$.

This body is subject to both a viscous torque that is proportional to the angular velocity of the body, \textit{i}.\textit{e}. of the form $C\bm{\omega}_R$, as well as a random torque (in the context of Brownian motion, this arises from molecular collisions) $\bm{\eta}$. The equation of motion is then the familiar Euler equation for angular momentum with random torque:
\begin{equation*}
    (\dot{\boldsymbol{\ell}} + I^{-1}\boldsymbol{\ell}\times\boldsymbol{\ell}) = -CI^{-1}\boldsymbol{\ell} + \bm{\eta}.
\end{equation*}
Writing this as a stochastic differential equation by noting that $\bm{\eta}\;dt = B\;d\bm{W}$, where $d\bm{W}$ is a Wiener process increment with zero mean and variance $dt$, we have,
\begin{equation}\label{eq:AngMomentSDE}
    (\dot{\boldsymbol{\ell}} + I^{-1}\boldsymbol{\ell}\times\boldsymbol{\ell})\;dt = -CI^{-1}\boldsymbol{\ell}\;dt + B\;d\bm{W}.
\end{equation}
Note that $B$ controls the variance and `color' of the noise term. 
In the non-inertial theory, the cross-product term $I^{-1}\boldsymbol{\ell}\times \boldsymbol{\ell}$ is assumed to be much smaller than the viscous torque $CI^{-1}\boldsymbol{\ell}$.
This is true when the viscosity of the fluid is relatively large and the particles are small.
We explain in Appendix I that these two terms can be of a comparable magnitude when the fluid is not viscous enough, \textit{e}.\textit{g}., in the air, where the inertial theory should apply.

In addition, since we are interested in obtaining the probability distribution over phase space $(R,\bm{\ell})$ we also have the following definition,
\begin{equation}\label{eq:RotSDE}
    (R^T\dot{R})^\vee dt = I^{-1}\bm{\ell}\;dt.
\end{equation}
Simultaneously solving (\ref{eq:AngMomentSDE}) and (\ref{eq:RotSDE}) we obtain the joint distribution $f(R,\boldsymbol{\ell};t)$ on phase space. In the following section, we describe the geometry of this phase space.

\subsection{Geometry of Phase Space}
From Hamiltonian dynamics, the phase space can be represented by the set of all orientation and angular momentum states that the body can occupy. We imbue an additional semi-direct product structure to this space so that we have a group $SO(3)\ltimes\mathbb{R}^3$ as the phase space on which the system evolves. From \cite{jayaraman2020black} we see that this is in fact the cotangent bundle group of $SO(3)$. A group element, $h(R,\bm{\ell}) \in SO(3)\ltimes\mathbb{R}^3$ is constructed as,
\begin{equation}\label{eq:h_def}
    h(R,\bm{\ell}) = \left(
    \begin{array}{c|c}
    R^T & \bm{\ell}\\
    \hline
    \bm{0}^T & 1
    \end{array}\right).
\end{equation}
Note that we parameterize rotations with $R^T$ instead of $R$ (nevertheless, this group is simply a different parameterization of $SE(3)$).

The Lie algebra of $SO(3)\ltimes\mathbb{R}^3$ is six-dimensional, and can be expressed using the same orthonormal basis as $se(3)$:
\begin{equation}\nonumber
    \tilde{E}_1 =
    \begin{pmatrix}
    0 && 0 && 0 && 0\\
    0 && 0 && -1 && 0\\
    0 && 1 && 0 && 0\\
    0 && 0 && 0 && 0
    \end{pmatrix},
    \;\;
    \tilde{E}_2 =
    \begin{pmatrix}
    0 && 0 && 1 && 0\\
    0 && 0 && 0 && 0\\
    -1 && 0 && 0 && 0\\
    0 && 0 && 0 && 0
    \end{pmatrix},
\end{equation}
\begin{equation}\nonumber
    \tilde{E}_3 =
    \begin{pmatrix}
    0 && -1 && 0 && 0\\
    1 && 0 && 0 && 0\\
    0 && 0 && 0 && 0\\
    0 && 0 && 0 && 0
    \end{pmatrix},
    \;\;
    \tilde{E}_4 =
    \begin{pmatrix}
    0 && 0 && 0 && 1\\
    0 && 0 && 0 && 0\\
    0 && 0 && 0 && 0\\
    0 && 0 && 0 && 0
    \end{pmatrix},
\end{equation}
\begin{equation}\nonumber
    \tilde{E}_5 =
    \begin{pmatrix}
    0 && 0 && 0 && 0\\
    0 && 0 && 0 && 1\\
    0 && 0 && 0 && 0\\
    0 && 0 && 0 && 0
    \end{pmatrix}
    \;\;\text{and}\;\;
    \tilde{E}_6 =
    \begin{pmatrix}
    0 && 0 && 0 && 0\\
    0 && 0 && 0 && 0\\
    0 && 0 && 0 && 1\\
    0 && 0 && 0 && 0
    \end{pmatrix}.
\end{equation}
The adjoint matrix of the group is $[Ad(h)]$ and is given by,
\begin{equation}
[Ad(h)] = 
    \begin{pmatrix}
    R^T && \mathbb{O}\\
    LR^T && R^T
    \end{pmatrix},
\end{equation}
where $L^\vee$ = $\boldsymbol{\ell}$. Moreover, the ``little ad" operator is given by $[ad(X)]$ where $X$ is in the Lie algebra of $SO(3)\ltimes\mathbb{R}^3$. If,
\begin{equation}
    X = 
    \begin{pmatrix}
    \Omega && \boldsymbol{v}\\
    \boldsymbol{0}^T && 0
    \end{pmatrix},
\end{equation}
then for $V^\vee = \boldsymbol{v}$,
\begin{equation}
    [ad(X)] = 
    \begin{pmatrix}
    \Omega && \mathbb{O}\\
    V && \Omega
    \end{pmatrix}.
\end{equation}
The right and left Jacobians, $\mathcal{J}_L$ and $\mathcal{J}_R$ can be calculated as,
\begin{equation}
    \mathcal{J}_r = 
    \begin{pmatrix}
    J_r(R^T) && \mathbb{O}\\
    \mathbb{O} && R
    \end{pmatrix}
    \;\;\text{and}\;\;
    \mathcal{J}_l = 
    \begin{pmatrix}
    J_l(R^T) && \mathbb{O}\\
    LJ_l(R^T) && \mathbb{I}
    \end{pmatrix},
\end{equation}
where $J_r$ and $J_l$ are the right and left Jacobians for $SO(3)$ given as,
\begin{equation}\nonumber
    J_r(R(\boldsymbol{q})) = \left[\left(R^T\frac{\partial R}{\partial q_1}\right)^\vee,\cdots,\left(R^T\frac{\partial R}{\partial q_3}\right)^\vee\right],
\end{equation}
\begin{equation}
    J_l(R(\boldsymbol{q})) = \left[\left(\frac{\partial R}{\partial q_1} R^T\right)^\vee,\cdots,\left(\frac{\partial R}{\partial q_3}R^T\right)^\vee\right].
\end{equation}
Using the definitions for the left and right Jacobians, we can construct the left and right Lie derivatives for this group (where the summation is implied over repeated indices),
\begin{align}\label{eq:ModSE(3)_right}
\tilde{E}^R_i u(h(\boldsymbol{q},\boldsymbol{\ell})) &= 
\begin{cases}
E^L_i u(h(\bm{q},\bm{\ell}))\; \; \text{$(i, j = 1, 2, 3)$}\\\\
R_{(i-3),j} \partial \tilde{u}/\partial \ell_j \; \; \text{$(j = 1, 2, 3$; $i = 4, 5, 6)$}
\end{cases}\\\label{eq:ModSE(3)_left}
\tilde{E}^L_i u(h(\boldsymbol{q},\boldsymbol{\ell})) &= 
\begin{cases}
E^R_i u(h(\bm{q},\bm{\ell})) - \varepsilon_{ijk}\ell_j\partial \tilde{u}/\partial l_k\;\; \text{$(i,j,k = 1, 2, 3)$}
\\\\
-\partial \tilde{u}/\partial \ell_{i-3} \; \; \text{$(i = 4, 5, 6)$}
\end{cases}
\end{align}
where $h$ is parameterized by $\boldsymbol{q}$ for the rotation and $\boldsymbol{\ell}$ for angular momentum and $u(h(\boldsymbol{q},\boldsymbol{\ell})) = \tilde{u}(\boldsymbol{q},\boldsymbol{\ell})$. Also, $\varepsilon_{ijk}$ is the Levi-Civita symbol. We use the formulae, $\boldsymbol{E}^R f = [J_r]^{-T}\partial u/\partial\boldsymbol{q}$ and $\boldsymbol{E}^L f = -[J_l]^{-T}\partial u/\partial\boldsymbol{q}$, and, $\boldsymbol{\tilde{E}}^R f = [\mathcal{J}_R]^{-T}\partial u/\partial\boldsymbol{\xi}$ and $\boldsymbol{\tilde{E}}^L f = -[\mathcal{J}_L]^{-T}\partial u/\partial\boldsymbol{\xi}$ where $\boldsymbol{\xi} = [\boldsymbol{q}^T,\boldsymbol{\ell}^T]^T$.

\subsection{Ornstein-Uhlenbeck Process in $SO(3)\ltimes\mathbb{R}^3$}

We observe that the left-hand side of the equations of motion (\ref{eq:AngMomentSDE},\ref{eq:RotSDE}) can be written as a component of $(\dot{h}h^{-1})^\vee dt$ where $h$ is given in (\ref{eq:h_def}). Since,
\begin{equation}
    (\dot{h}h^{-1})^\vee =
    \begin{pmatrix}
    \dot{R}^TR && -\dot{R}^TR\boldsymbol{\ell} + \dot{\boldsymbol{\ell}}\\
    \boldsymbol{0}^T && 0
    \end{pmatrix}^\vee.
\end{equation}
and $(\dot{R}^TR)^\vee = -(R^T\dot{R})^\vee = -\boldsymbol{\omega}_R$. This also implies that $-\dot{R}^TR\boldsymbol{\ell} = \bm{\omega}_R\times\bm{\ell}$ and we have,
\begin{equation}
    (\dot{h}h^{-1})^\vee = 
    \begin{pmatrix}
    -\boldsymbol{\omega}_R\\
    \dot{\boldsymbol{\ell}} + \boldsymbol{\omega}_R\times\boldsymbol{\ell}
    \end{pmatrix}.
\end{equation}
Comparing with (\ref{eq:AngMomentSDE},\ref{eq:RotSDE}) we see that,
\begin{equation}\label{eq:mainSDE}
    (\dot{h}h^{-1})^\vee dt = -
    \begin{pmatrix}
    I^{-1}\boldsymbol{\ell}\\
    CI^{-1}\boldsymbol{\ell}
    \end{pmatrix}\; dt +
    \begin{pmatrix}
    \mathbb{O} && \mathbb{O}\\
    \mathbb{O} && B
    \end{pmatrix}\circledS\;d\boldsymbol{w},
\end{equation}
which is a (left) stochastic differential equation evolving on $SO(3)\ltimes\mathbb{R}^3$; the $\circledS$ makes it explicit that this is a Stratonovich equation. Notably, the troublesome cross-product term is absorbed into the definition of $h$, therefore, avoiding the need to deal with that term explicitly. Setting 
\begin{equation}
\boldsymbol{m} = 
\begin{pmatrix}
I^{-1}\boldsymbol{\ell}\\
CI^{-1}\boldsymbol{\ell}
\end{pmatrix}
\;\;\text{and}\;\;
\tilde{B} = 
\begin{pmatrix}
\mathbb{O} && \mathbb{O}\\
\mathbb{O} && B
\end{pmatrix},
\end{equation}
the Fokker-Planck equation for the probability density function $u(h,t)$ on $SO(3)^T\ltimes\mathbb{R}^3$ is,
\begin{equation}\label{eq:FPE_OU_Group}
    \frac{\partial u}{\partial t} =- \sum_{i = 1}^6 E^L_i(m_i u) + \frac{1}{2}\sum_{i,j = 1}^6 (\tilde{B}\tilde{B}^T)_{ij}E^L_iE^L_j u.
\end{equation}
The procedure to go from a stochastic differential equation evolving on a group to a corresponding Fokker-Planck equation is provided in detail in \cite{chirikjian2011stochastic}. In coordinates $\boldsymbol{\xi} = [\boldsymbol{q}^T,\boldsymbol{\ell}^T]^T$, we express the equation in terms of $\tilde{u}(\boldsymbol{\xi},t) = u(h(\boldsymbol{\xi}),t)$ as,
\begin{equation}\label{eq:FPE_OU_Group_Coord}
    \frac{\partial \tilde{u}}{\partial t} = -\sum_{i,j = 1}^6 Q_{ij}\tilde{E}^L_i(\xi_j \tilde{u}) + \frac{1}{2}\sum_{i,j = 1}^6 (\tilde{B}\tilde{B}^T)_{ij}\tilde{E}^L_i\tilde{E}^L_j \tilde{u},
\end{equation}
where $Q$ is,
\begin{equation}
    Q = 
    \begin{pmatrix}
    \mathbb{O} && I^{-1}\\
    \mathbb{O} && CI^{-1}
    \end{pmatrix},
\end{equation}
therefore describing a degenerate Ornstein-Uhlenbeck process. The matrix $Q$ has a maximum rank of 3, and therefore only 3 non-zero singular values.

\subsection{Maxwell-Boltzmann Limit}\label{sec:MaxwellBoltzmann}
The limiting distribution at $t\rightarrow \infty$ is a Maxwell-Boltzmann distribution for angular momentum and a uniform distribution over orientation, \textit{i}.\textit{e}.,
\begin{equation}\label{eq:LimitingSol}
    u_{\infty}(R,\boldsymbol{\ell}) = \frac{1}{Z}\exp\left(-\frac{1}{2}\beta \boldsymbol{\ell}^TI^{-1}\boldsymbol{\ell}\right),
\end{equation}
which is known from thermodynamics. Here, $\beta = 1/k_BT$ where $k_B$ is the Boltzmann constant, and $Z$ is the partition function to normalize the distribution over phase space. To ensure that this limiting distribution can be obtained we also need,
\begin{equation}\label{eq:Fluct_Diss}
    C+C^T = \beta BB^T,
\end{equation}
by the fluctuation-dissipation theorem. 
A proof using a non-parametric approach is provided in Appendix II.
For usual physical systems, this reduces to,
\begin{equation}\label{eq:Fluct_Diss_Phy}
2C = \beta BB^T,
\end{equation}
because $C$ is a symmetric matrix as per Onsager's reciprocity theorem \cite{onsager1931reciprocal,coleman1960reciprocal}. 
Note that a similar relationship between a diffusion matrix and viscous matrix is derived in \cite{brenner1967coupling}.
However, their diffusion matrix $D_0$ has a different physical meaning from the $BB^T$ we are using here.

\section{Solution Methodology and Results}
To approximately solve the Fokker-Planck equation in (\ref{eq:FPE_OU_Group_Coord}) we adapt the solution provided by Riksen in \cite{risken1996fokker}. We consider the initial conditions of $\tilde{u}(h,0) = \delta_G(h)$ where $\delta_G(h)$ is the group Dirac delta function concentrated at the identity element (this implies that the particle has zero angular velocity at $t = 0$ and all orientations are defined with respect to the orientation at $t = 0$). If the probability distribution is sufficiently concentrated we can make the following approximation,
\begin{equation} \label{eq:ori_assump}
    \tilde{E}^L_i \approx - \frac{\partial}{\partial \xi_i},
\end{equation}
so that we have $u(h,t) = \tilde{u}(\bm{\xi},t)$ for sufficiently small time $t$,
\begin{equation}\label{eq:FPE_Approximate_Xi}
    \frac{\partial \tilde{u}}{\partial t} = \sum_{i,j = 1}^6 Q_{ij}\frac{\partial}{\partial \xi_i}(\xi_j \tilde{u}) + \frac{1}{2}\sum_{i,j = 1}^6 (\tilde{B}\tilde{B}^T)_{ij}\frac{\partial^2\tilde{u}}{\partial\xi_i\partial\xi_j},
\end{equation}
which solves to yield a Gaussian with mean $\bm{\mu}(t) = \bm{0}$ and covariance $\Sigma(t)$ that varies as,
\begin{equation}\label{eq:Sigma_Eq}
    \Sigma(t) = \int_0^t [e^{-Q\tau}]\tilde{B}\tilde{B}^T[e^{-Q\tau}]^T\;d\tau.
\end{equation}
Assuming that this is the solution on the Lie algebra we can construct the solution on the group as,
\begin{equation}\label{eq:Solution}
u(h,t) \!=\! \frac{1}{(2\pi)^3|\text{det}\;\Sigma(t)|^{\frac{1}{2}}}\exp\!\left(\!-\frac{1}{2}[\log^\vee h]^T \Sigma^{-1}(t)[\log^\vee h]\!\right),
\end{equation}
where $\log$ is the matrix logarithm. Hence, we obtain a closed-form approximate solution for the joint probability distribution evolving in phase space. We also note that the group theoretic mean and covariance are defined as,
\begin{equation}\label{eq:GroupMu}
\int_G \log^\vee(h \circ \mu^{-1})\;u(h,t)\;dh \;\doteq\; \boldsymbol{0},
\end{equation}
and,
\begin{equation}\label{eq:GroupCov}
\Sigma(t) \;\doteq\; \int_G [\log^\vee(h \circ \mu^{-1})][\log^\vee(h \circ \mu^{-1})]^T\;u(h,t)\;dh.
\end{equation}


\subsection{Example: Diagonal and Isotropic Viscous Tensor}
Consider a body fixed frame that is oriented along the eigenvectors of the moment of inertia tensor $I$ so that we have a diagonal moment of inertia tensor, $I = \text{diag}(I_1,I_2,I_3)$. Then assume that the viscous tensor is of the form $C = c\mathbb{I}$ (i.e., diagonal and isotropic). The $\tilde{B}$ is also given as,
\begin{equation}
    \tilde{B} = 
    \begin{pmatrix}
    \mathbb{O} && \mathbb{O}\\
    \mathbb{O} && b\mathbb{I}
    \end{pmatrix},
\end{equation}
where $2c = \beta b^2$ from (\ref{eq:Fluct_Diss}). Then, 
\begin{equation}
    Q = \begin{pmatrix}
    \mathbb{O} && \text{diag}(1/I_1,1/I_2,1/I_3)\\
    \mathbb{O} && \text{diag}(c/I_1,c/I_2,c/I_3)
    \end{pmatrix},
\end{equation}
so that,
\begin{equation}\label{eq:CovSol}
    \Sigma(t) = \int_{0}^t [e^{-Q\tau}]\tilde{B}\tilde{B}^T[e^{-Q\tau}]^T\;d\tau\\
    = \begin{pmatrix}
    \sigma_{11}(t) && \sigma_{12}(t)\\
    \sigma_{21}(t) && \sigma_{22}(t)
    \end{pmatrix},
\end{equation}
where,
\begin{align*}
    \sigma_{11}(t) &= b^2\text{diag}(\sigma_{11}(I_1;t),\sigma_{11}(I_2;t),\sigma_{11}(I_3;t))\\
    \sigma_{12}(t) &= b^2\text{diag}(\sigma_{12}(I_1;t),\sigma_{12}(I_2;t),\sigma_{12}(I_3;t))\\
    \sigma_{21}(t) &= b^2\text{diag}(\sigma_{21}(I_1;t),\sigma_{21}(I_2;t),\sigma_{21}(I_3;t))\\
    \sigma_{22}(t) &= b^2\text{diag}(\sigma_{22}(I_1;t),\sigma_{22}(I_2;t),\sigma_{22}(I_3;t)),
\end{align*}
such that,
\begin{align*}\label{eq:Sigma_R}
    &\sigma_{11}(I_i;t) = \frac{b^2}{c^2}\left(t - \frac{3I_i}{2c}- \frac{I_i}{2c}(e^{-2ct/I_i} - 4e^{-ct/I_i})\right)\\
    &\sigma_{12}(I_i;t) = \sigma_{21}(I_i;t) = \frac{b^2}{c}\left(\frac{I_i}{c}e^{-ct/I_i} - \frac{I_i}{2c}e^{-2ct/I_i} - \frac{I_i}{2c}\right)\\
    &\sigma_{22}(I_i;t) = -\frac{I_i b^2}{2c}(e^{-2ct/I_i}-1).
\end{align*}
We see that when we set $I_i = 0$ for $i = 1, 2, 3$, we obtain $\sigma_{11}(t) = b^2 t/c^2$ and $\sigma_{12}(t) = \sigma_{21}(t) = \sigma_{22}(t) = 0$. In fact, this is the solution obtained from a non-inertial theory as will be explored in the following section.
\subsubsection{Comparison against non-inertial theory}
In the non-inertial theory of rotational Brownian motion, the governing stochastic differential equation is of the form,
\begin{align}\nonumber
    C (R^T\dot{R})^\vee dt &= Bd\boldsymbol{W}.
\end{align}
This yields the following Fokker-Planck equation for $u(R,t)$,
\begin{equation}\label{eq:noninertial_FPE}
    \frac{\partial u}{\partial t} =  \frac{1}{2}\sum_{i,j=1}^3 D_{ij}E^R_iE^R_j u,
\end{equation}
for $D = C^{-1}BB^TC^{-T}$, which is a drift-less diffusion on $SO(3)$. If we consider the solution at small times (which can be approximated via covariance propagation \cite{chirikjian2016harmonic}) and an initial condition of zero initial angular velocity, we obtain the following approximate solution:
\begin{align}\nonumber
    u(R,t) &= \frac{1}{(2\pi)^{3/2}|Dt|^\frac{1}{2}}\exp\left(-\frac{1}{2t}[\log^\vee R]^TD^{-1}[\log^\vee R]\right).
\end{align}
Noting that the covariance terms corresponding to the three-dimensional vector $\log^\vee R$ is $\sigma_{11}$ from (\ref{eq:CovSol}), we proceed to make a comparison with the results from the theory of inertial rotational Brownian motion. 
Now, if $B = b\mathbb{I}$ is the identity matrix and $C = c\mathbb{I}$, this gives $\Sigma = Dt = b^2t/c^2\mathbb{I}$. However, by setting $I_1 = I_2 = I_3 = 0$ in (\ref{eq:CovSol}), we obtain the same form for the orientational covariance $\sigma_{11}(t)$ from (\ref{eq:CovSol}).

\subsubsection{Spherical particles}
As a special case, we consider a spherical particle, where $I_1 = I_2 = I_3$, to demonstrate the underlying ideas.

As ground truth, we compare against the covariance (the mean is always identity) obtained from an It\^{o}-Gangolli integration of the governing stochastic differential equation in (\ref{eq:mainSDE}). That is, we evolve an ensemble of $i = 1,\cdots,N$ particles on the group with time as follows with the following update over a time step $dt$:
\begin{equation}\label{eq:ItoGangolliInjection}
  h_i(t + dt) = \exp\left(dX\right)\circ h_i(t),
\end{equation}
where $dX^\vee$ is the right-hand side of (\ref{eq:mainSDE}). A probability distribution can be constructed from these sample points as,
\begin{equation}
    u_s(h,t) = \frac{1}{N}\sum_{i=1}^N\delta\left(h\circ h_i^{-1}(t)\right),
\end{equation}
and the sample covariance is,
\begin{equation}
    \Sigma_s(t) = \frac{1}{N}\sum_{i=1}^N[\log^\vee h_i(t)][\log^\vee h_i(t)]^T.
\end{equation}
The above assumes that the sample mean is identity at all times. This is seen as a reasonable approximation since numerically, the sample mean deviates from the identity with only a 3\% error (measured by a relative Frobenius norm) for $N = 30,000$ particles. 

In what follows, we assume $N = 30,000$ particles and compare the covariance obtained from the Ornstein-Uhlenbeck solution in (\ref{eq:CovSol}) with the sampled covariance for $t = 3$ units (and time step of $dt = 1\times 10^{-3}$ units). Additionally, $I = C = \mathbb{I}$ and $\beta = 0.2$ (the fluctuation-dissipation theorem sets the value of $b = \sqrt{2/\beta}$). Furthermore, we decompose the $6\times6$ covariance matrix into blocks of $3\times3$ matrices as,
\begin{equation}
    \Sigma = 
    \begin{pmatrix}
    \Sigma_{RR} && \Sigma_{R\ell}\\
    \Sigma_{R\ell}^T && \Sigma_{\ell\ell}
    \end{pmatrix},
\end{equation}
where $\Sigma_{RR}$ and $\Sigma_{\ell\ell}$ are related to the covariance of the distribution once the distribution is marginalised over angular momenta and orientation respectively; $\Sigma_{R\ell}$ is the covariance for the coupling of the two variables. 

The error between the covariance matrix obtained from the Ornstein-Uhlenbeck solution in (\ref{eq:CovSol}) and the sample covariance matrix is computed as,
\begin{equation}
    \text{Error}(t) = \frac{||\Sigma(t) - \Sigma_s(t)||_F}{||\Sigma_s(t)||_F},
\end{equation}
where $||\cdots||_F$ is the Frobenius norm. The evolution of this relative error with time is plotted in Figure \ref{fig:FrobeniusErrorSigma}. Notably, the error is small (order of 1\%) till $t\approx 0.6$ units after which the error grows monotonically, and the approximation is no longer valid. 
\begin{figure}[!ht]
\includegraphics[width = 0.47\textwidth]{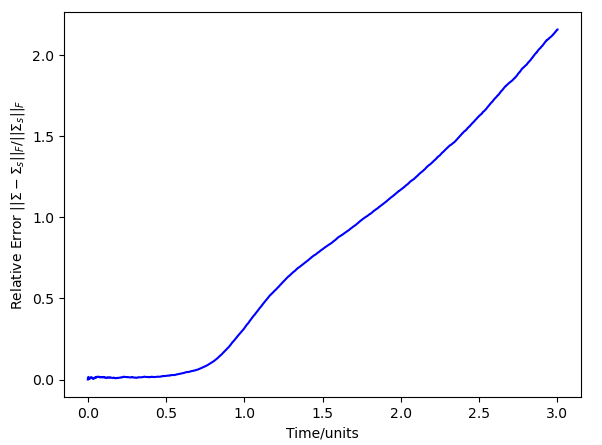}
\caption{\label{fig:FrobeniusErrorSigma} The evolution of the relative error with time for the Ornstein-Uhlenbeck solution, compared against the sample covariance.}
\end{figure}

Since for a sphere, all three covariance matrices can be expected to be multiples of $\mathbb{I}$ (from (\ref{eq:CovSol})), we now seek to compare the time-evolution for (co)variance $\Sigma_{RR}$, $\Sigma_{R\ell}$ and $\Sigma_{\ell\ell}$; that is, for a covariance $\Sigma = k\mathbb{I}$, we compare in terms of the scalar $k$.
For the sample covariance, we use the average of diagonal elements as $k$.
\paragraph{Covariance for orientations} If the probability density for orientation for a non-inertial rotational Brownian motion is solved by covariance propagation, we obtain $\Sigma_{RR} = tD\mathbb{I}$ where $D = b^2/c^2$. 
We also plot this solution along with that from the Ornstein-Uhlenbeck approximation (for inertial theory). Finally, from Section \ref{sec:MaxwellBoltzmann} we know that the limiting orientational distribution is a uniform distribution in $SO(3)$, i.e., $u_\infty(R) = 1/(8\pi^2)$ so that we have,
\begin{equation*}
    \int_{SO(3)}\frac{1}{8\pi^2}\;dg = \frac{1}{8\pi^2}\int_0^{2\pi}\int_0^\pi\int_0^{2\pi}\sin\beta\;d\alpha d\beta d\gamma = 1,
\end{equation*}
where the integral over $SO(3)$ is decomposed to an integral over the $ZYZ$ Euler angles $(\alpha,\beta,\gamma)$. The limiting covariance is then,
\begin{equation*}
    \Sigma_{RR,\infty} = \frac{1}{8\pi^2}\int_G[\log^\vee g][\log^\vee g]^T\;dg.
\end{equation*}
Now we parameterize $g = g(\boldsymbol{q})$ in terms of exponential coordinates, i.e., $g = \exp(\boldsymbol{q}^\wedge)$. In this parameterisation, $$dg = \frac{2(1 - \cos||\boldsymbol{q}||)}{||\boldsymbol{q}||^2}\;d\boldsymbol{q}.$$ Since $\boldsymbol{q}\in\mathbb{R}^3$, we imagine decomposing the integral into an integration over a solid ball of radius $r = ||\boldsymbol{q}||$ where $0\leq r < \pi$. Then, letting the polar angle be $\theta$, for $0\leq \theta < \pi$, and the azimuthal angle be $\phi$, for $0\leq \phi < 2\pi$, we can write, $\boldsymbol{q} = [r\cos\phi\sin\theta,r\sin\phi\sin\theta,r\cos\theta]^T$. Thus,
\begin{align}\nonumber
    \Sigma_{RR,\infty} &= \frac{1}{8\pi^2}\int_0^\pi\int_0^{2\pi}\int_0^{\pi}\boldsymbol{q}\boldsymbol{q}^T\;2(1 - \cos r)\sin\theta\;d\theta d\phi dr\\\label{eq:LimitingCovarianceR}
    &= \frac{\pi^2 + 6}{9}\mathbb{I}.
\end{align}
Figure \ref{fig:SRR_sphere} is a plot comparing the (co)variance obtained from the Ornstein-Uhlenbeck solution in (\ref{eq:CovSol}) along with the corresponding term in the sample covariance. Additionally, the covariance from the non-inertial theory and that for the limiting uniform distribution (\ref{eq:LimitingCovarianceR}) are plotted. 
\begin{figure}[!ht]
\includegraphics[width = 0.45\textwidth]{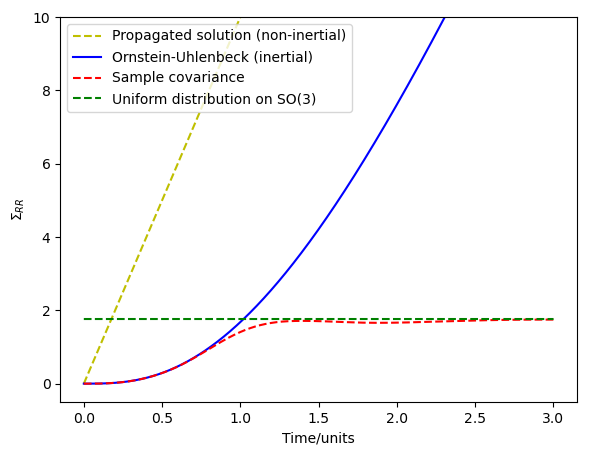}
\caption{\label{fig:SRR_sphere} The evolution of the (co)variance for orientations in the Ornstein-Uhlenbeck solution, compared against that for the sample covariance, along with the corresponding covariance obtained from the non-inertial theory and the uniform distribution asymptote.}
\end{figure}

We note that the covariance from the Ornstein-Uhlenbeck solution (\ref{eq:CovSol}) begins to deviate from the sample covariance after about $t\approx 0.6$ units. Numerically, this is when the covariance becomes large enough that a distribution in the $\mathbb{R}^3$ space spanned by the exponential parameterisation of orientation has sufficiently large volume outside the solid ball, $||\log^\vee g|| < \pi$ for $g \in SO(3)$. The curvature of $SO(3)$, therefore, makes the small covariance assumption that motivated the Ornstein-Uhlenbeck solution invalid. This is true even of the non-inertial theory in Figure \ref{fig:SRR_sphere} and thus, none of the propagation or Ornstein-Uhlenbeck solutions are able to recover the limiting uniform distribution (corresponding to thermal equilibrium).

\paragraph{Covariance for orientation-angular momentum coupling} Figure \ref{fig:SRL_sphere} depicts the time evolution of the coupling covariance, comparing the Ornstein-Uhlenbeck solution with the corresponding sample covariance.
\begin{figure}[!ht]
\includegraphics[width = 0.45\textwidth]{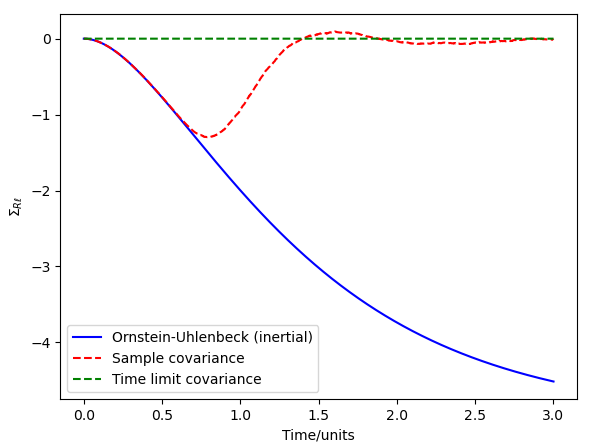}
\caption{\label{fig:SRL_sphere} The evolution of the (co)variance for the orientation-angular momentum coupling for the Ornstein-Uhlenbeck (OU) solution against the sample covariance and the time limit covariance.}
\end{figure}
At large times, the coupling (co)variance goes to zero, suggesting the decoupling between angular momentum and orientation, which is seen from the limiting Maxwell-Boltzmann solution in (\ref{eq:LimitingSol}).
\paragraph{Covariance for angular momentum} In exponential coordinates, $$\log^\vee h = \begin{pmatrix}
\boldsymbol{q}\\
J_l^{-1}(\boldsymbol{q})\boldsymbol{\ell}
\end{pmatrix},$$
where $J_l$ is the left Jacobian for $SO(3)$ for exponential coordinates \cite{chirikjian2011stochastic}. Thus $\Sigma_{\ell\ell}$ is taken to refer to the (co)variance associated with the vector $\boldsymbol{y} = J_l^{-1}(\boldsymbol{q})\boldsymbol{\ell}$ instead. 
Using the Maxwell-Boltzmann distribution and the definition (\ref{eq:GroupCov}), we can obtain that the time limit covariance for $\boldsymbol{y}$ is $\Sigma_{\ell\ell,\infty}=\frac{\pi^2+3}{9\beta} I$.
Figure \ref{fig:SlL_sphere} depicts the time evolution for the covariance associated with angular momentum, comparing the Ornstein-Uhlenbeck solution with the corresponding sample covariance and the time limit covariance.
\begin{figure}[!ht]
\vspace{0cm}
\includegraphics[width = 0.45\textwidth]{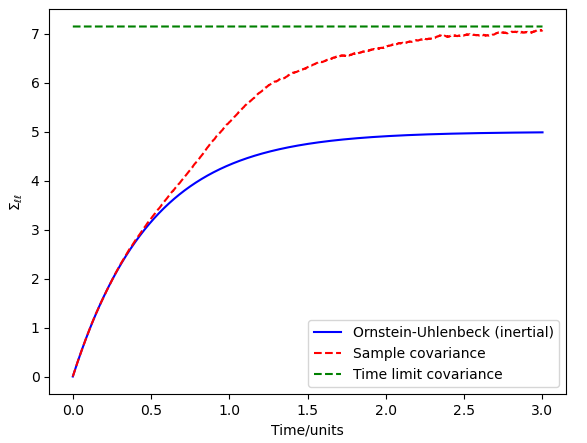}
\caption{\label{fig:SlL_sphere} The evolution of the (co)variance associated with the angular momentum for the Ornstein-Uhlenbeck (OU) solution against the sample covariance and the time limit covariance.}
\vspace{-0.5cm}
\end{figure}

\section{\label{sec:conclusion}Conclusions}

We have provided a short-time closed-form expression of the probability density function on the phase space for an arbitrary-shaped particle undergoing an inertial rotational Brownian motion.
The phase space is described by the cotangent bundle group of $SO(3)$.
An example, where the parameters of the PDF have an analytic expression, is shown and compared to the ground truth and the approximate solution of the non-inertial theory.
For more general cases, it is possible to use a numerical method to obtain the parameters and utilize the simple closed-form PDF to perform further calculations.

One possible future work is to fit a spline to connect the small-time solution with the steady-state Boltzmann distribution.
It requires an additional estimation of the relaxation time.
Another extension is to calculate the covariance matrix without using the approximation (\ref{eq:ori_assump}).
A possible approach is to estimate the time derivative of the covariance matrix (\ref{eq:GroupCov}) using the Fokker-Planck equation on phase space (\ref{eq:FPE_OU_Group_Coord}) and the Lie derivative approximations in Appendix III.

{\bf Acknowledgements}
This work was supported by NUS Startup  grants A-0009059-02-00 and A-0009059-03-00, National Research Foundation, Singapore, under its Medium Sized Centre Programme - Centre for Advanced Robotics Technology Innovation (CARTIN),  sub award A-0009428-08-00, and AME Programmatic Fund Project MARIO A-0008449-01-00.

\vspace{-0.4cm}
\newpage
\vspace{-0.4cm}
\bibliography{References}

\clearpage

\section*{APPENDIX I}\label{sec:App1}

In this section, we give a real-world example where the viscous torque is comparable to the torque introduced by inertia.
The modeling details and numerical results are presented.

Assume many Tobacco Mosaic Virus (TMV) particles are floating and rotating in the air subject to viscous torque and a random torque.
A TMV particle is a long, stick-like particle, whose total length, inner, and outer diameter are $300\,nm$, $4\,nm$, and $18\,nm$ respectively.
It is composed of coat proteins and a single-strand RNA molecule inside, with a total particle weight $3.94\times 10^7\,Da$ \cite{bruckman2014chemical}.
We model a TMV particle as a solid ellipsoid with a major axis $a=150\,nm$, two minor axes $b=c=9\,nm$ and mass $M=3.94\times 10^7 \, Da$.
The viscous torque for a prolate ellipsoid is known \cite{jeffery1915steady},
\begin{equation} \label{vis_ellips}
\boldsymbol{T}^{hyd}=\frac{16}{3}\pi\eta d^3  \boldsymbol{\omega}\bigg / \left[ \frac{1}{2}\ln{\frac{a+d}{a-d}}-\frac{ad}{b^2}\right],
\end{equation}
where $\eta$ is the shear viscosity of the fluid and $d=\sqrt{a^2-b^2}$.
The inertia tensor of an ellipsoid is (when the frame is aligned with the principle axes):
\begin{equation} \label{eq:inertia_ellips}
    I=\frac{M}{5}\begin{pmatrix}
    b^2+c^2 & 0 & 0\\
    0 & a^2+c^2 & 0 \\
    0 & 0 & a^2+b^2
    \end{pmatrix},
\end{equation}
where $M$ is the total mass.

The long-term distribution of the angular velocity of these particles should be a Maxwell-Boltzmann distribution,
\begin{equation} \label{eq:p_omega}
    \tilde{u}_{\infty}(\boldsymbol{\omega}) = \frac{1}{Z}\exp\left(-\frac{1}{2k_B T} \boldsymbol{\omega}^TI \boldsymbol{\omega}\right),
\end{equation}
following the distribution of the angular momentum (\ref{eq:bolt_lim}).
\begin{figure}[!ht]
\includegraphics[width = 0.4\textwidth]{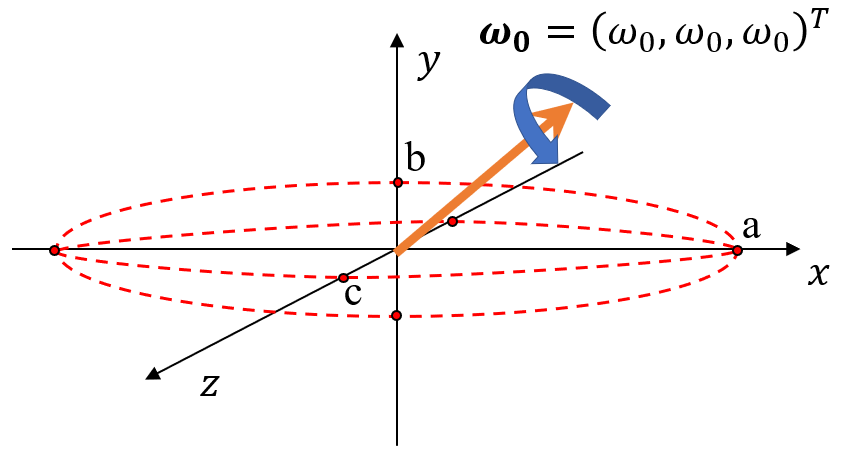}
\caption{\label{fig:omega_illus} The illustration of a prolate ellipsoid rotating at the angular velocity  $\boldsymbol{\omega}_0=({\omega}_0,{\omega}_0,{\omega}_0)^T$.}
\end{figure}
If an angular velocity $\boldsymbol{\omega}_0$ is relatively ``usual'' for the distribution (\ref{eq:p_omega}), the value $\frac{1}{2k_B T} \boldsymbol{\omega}_0^TI \boldsymbol{\omega}_0$ should not be too large.
Assume an angular velocity takes the following form,
\begin{equation} \label{commom_omega}
    \boldsymbol{\omega}_0 = (\omega_0,\omega_0,\omega_0)^T,
\end{equation}
which is illustrated in Figure \ref{fig:omega_illus}.
Assume the temperature $T=300K$. 
For a ``usual'' angular velocity $\boldsymbol{\omega}_0$ that satisfies:
\begin{equation}
    \frac{1}{2k_B T} \boldsymbol{\omega}_0^TI \boldsymbol{\omega}_0
=2,
\end{equation}
we can obtain,
\begin{equation}
    \omega_0=\sqrt{2\cdot 2k_B T / tr(I)}=5.29\times 10^6 rad/s
\end{equation}
using (\ref{eq:inertia_ellips}) and the parameters of a TMV particle.
At this temperature, the shear viscosity of air is $\eta=1.86\times 10^{-5}\, Pa\cdot s$. 
The viscous torque using (\ref{vis_ellips}) is,
\begin{equation}
    C\boldsymbol{\omega}_0=  \begin{pmatrix}
        -2.02\times 10^{-20} \\ -2.02\times 10^{-20} \\ -2.02\times 10^{-20}
    \end{pmatrix} N\cdot m,
\end{equation}
and the torque introduced by inertia is,
\begin{equation}
    \boldsymbol{\omega}_0\times (I\boldsymbol{\omega}_0)=
\begin{pmatrix}
        -7.10\times 10^{-37} \\ -8.19\times 10^{-21} \\ +8.19\times 10^{-21}
    \end{pmatrix} N\cdot m.
\end{equation}
In this case, we see the viscous and inertial torques are of a comparable magnitude.
So when the distribution is approaching the limit distribution, a considerable amount of particles are influenced by both torques, where the inertial theory becomes important.

\section*{APPENDIX II}\label{sec:App2}

In this section, we provide a proof for the relationship (\ref{eq:Fluct_Diss}) in a more general setting, where a deterministic force field is also present.
The proof is written for inertial rotational Brownian motion, while the proof for the non-inertial case follows the same way.
Below, we use the ``$\wedge$'' and ``$\vee$'' operators, the basis of Lie algebra, and Lie derivative $E^R_k$ of $SO(3)$ defined in \cite{chirikjian2011stochastic}.
The Einstein summation convention is also used to simplify equations.

Assume a particle is rotating in a deterministic field $V(R)$ with linear viscous torque and random torque. 
The following stochastic differential equations on the direct product group $SO(3)\times \mathbb{R}^3$ describes the movement of the particle:
\begin{equation}
    \begin{cases}
    d\boldsymbol{\ell}=(\hat{\boldsymbol{\ell}} I^{-1}\boldsymbol{\ell}-E^R V - CI^{-1}\boldsymbol{\ell})dt+Bd\boldsymbol{w}, \\
    (R^T dR)^{\vee}=I^{-1}\boldsymbol{\ell} dt,
    \end{cases}
\end{equation}
where $\boldsymbol{\ell}=I\boldsymbol{\omega}$ is the angular momentum. 
We also rewrite the cross-product term $\boldsymbol{\ell}\times(I^{-1}\boldsymbol{\ell})$ as a matrix multiplication term $\hat{\boldsymbol{\ell}}I^{-1}\boldsymbol{\ell}$ using the property of skew-symmetric matrices. 
The Lie derivative term stands for a vector $E^R V=(E^R_1 V,\, E^R_2 V,\, E^R_3 V)^T$.
Using the methodology in \cite{chirikjian2011stochastic}, the corresponding Fokker-Planck equation on $SO(3)\times \mathbb{R}^3$ is:
\begin{equation}
    \frac{\partial f}{\partial t}=-E^R_k(h_k^Rf)-\frac{\partial}{\partial \ell_k}(h^{\ell}_k f)+\frac{1}{2}(B{B}^T)_{kl}\frac{\partial ^2 f}{\partial \ell_k \partial \ell _l} \label{eq:FPE_direct},
\end{equation}
where
\begin{equation}
\begin{cases}
 \boldsymbol{h}^{\ell} =\hat{\boldsymbol{\ell}}I^{-1}\boldsymbol{\ell}-E^RV-CI^{-1}\boldsymbol{\ell},\\
 \boldsymbol{h}^R =I^{-1}\boldsymbol{\ell}.
\end{cases}
\end{equation}
Here, we use the notation $f(R,\boldsymbol{\ell},t)$ to denote the distribution on the direct product group, so as to differentiate from the distribution $u(h,t)$ on the cotangent bundle in the main text.
We proceed to prove the main theorem.
\begin{theorem}
The solution $f(R,\ell,t)$ to the Fokker-Planck equation (\ref{eq:FPE_direct}) satisfies
\begin{equation} \label{eq:bolt_lim}
    \lim_{t\to \infty} f(R,\boldsymbol{\ell},t)=\frac{1}{Z}\exp \left(-\beta \left( \frac{1}{2}\boldsymbol{\ell}^TI^{-1}\boldsymbol{\ell}+V(R) \right) \right),
\end{equation}
if and only if,
\begin{equation}\label{bbt_cond}
C+C^T=\beta BB^T.
\end{equation}
\end{theorem}
\begin{proof}
When the time is very large, the LHS of (\ref{eq:FPE_direct}) should approach zero.
We denote the RHS of (\ref{eq:bolt_lim}) as $f_{\infty}(R,\boldsymbol{\ell})$ and substitute it into the RHS of (\ref{eq:FPE_direct}).
The first term is,
\begin{equation} \label{eq:1term}
    -E_k^R(h_k^Rf_{\infty})=-\boldsymbol{e}_k^TI^{-1}\boldsymbol{\ell}E^R_k f_{\infty},
\end{equation}
where $\boldsymbol{e}_k^TI^{-1}\boldsymbol{\ell}$ is a scalar and the subscript $k$ is summed over.
For the second term, we first calculate ${\partial h_k^{\ell}}/{\partial \ell _k}$,
\begin{equation}
\begin{aligned} \label{eq:h_der}
\frac{\partial h^{\ell}_k}{\partial {\ell}_k}&=\frac{\partial }{\partial {\ell}_k} \left[ \boldsymbol{e}^T_k(\hat{\boldsymbol{\ell}}I^{-1}\boldsymbol{\ell}-E^RV-CI^{-1}\boldsymbol{\ell}) \right] \\
&=\boldsymbol{e}^T_k(\hat {\boldsymbol{e}}_kI^{-1}\boldsymbol{\ell}+\hat{\boldsymbol{\ell}} I^{-1} \boldsymbol{e}_k-CI^{-1} \boldsymbol{e}_k) \\ 
&=\boldsymbol{e}_k^T(\hat{\boldsymbol{\ell}} -C)I^{-1} \boldsymbol{e}_k.  
\end{aligned} 
\vspace{0cm}
\end{equation}
Note that the term $\boldsymbol{e}^T_k \hat {\boldsymbol{e}}_kI^{-1}\boldsymbol{\ell}$ vanishes because $\boldsymbol{e}^T_k \hat {\boldsymbol{e}}_k=-(\boldsymbol{e}_k \times {\boldsymbol{e}}_k)^T$.
The derivatives of $f_{\infty}$ are: 
\begin{equation} \label{eq:f_der}
\begin{cases}
E^R_kf_{\infty}=-\beta (E^R_k V) f_{\infty} \\
\frac{\partial f_{\infty}}{\partial \ell_k}=-\beta \boldsymbol{e}_k^T I^{-1}\boldsymbol{\ell} f_{\infty} \\
\frac{\partial ^2 f_{\infty}}{\partial \ell_k \partial \ell_l}=\left(-\beta \boldsymbol{e}_k^TI^{-1}\boldsymbol{e}_l+\beta^2 \boldsymbol{\ell}^T(I^{-T}\boldsymbol{e}_k \boldsymbol{e}_l^TI^{-1}\right)\boldsymbol{\ell} )f_{\infty} 
\end{cases}.
\end{equation}
The fact that $\boldsymbol{e}_k^TI^{-1}\boldsymbol{\ell}=\boldsymbol{\ell}^TI^{-T}\boldsymbol{e}_k$ is used above.

Substituting (\ref{eq:1term}), (\ref{eq:h_der}), and (\ref{eq:f_der}) into (\ref{eq:FPE_direct}) and removing $f_{\infty}$, we have:
\begin{equation} \label{long}
    \begin{aligned}
    \beta (E^R_k V)\boldsymbol{e}_k^TI^{-1}\boldsymbol{\ell}-\boldsymbol{e}_k^T(\hat {\boldsymbol{\ell}} -C)I^{-1}\boldsymbol{e}_k + \beta (\boldsymbol{e}_k^T I^{-1}\boldsymbol{\ell}) h_k^{\ell} \\ +\frac{1}{2}(BB^T)_{kl}(-\beta \boldsymbol{e}_k^TI^{-1}\boldsymbol{e}_l+\beta^2\boldsymbol{\ell}^T(I^{-T}\boldsymbol{e}_k \boldsymbol{e}_l^TI^{-1})\boldsymbol{\ell})=0  .
    \end{aligned}
\end{equation}
The first three terms can be organized as:
\begin{equation} \label{long_simplify}
    \begin{aligned}
&\beta (E^R_k V)\boldsymbol{e}_k^TI^{-1}\boldsymbol{\ell}-\boldsymbol{e}_k^T(\hat{\boldsymbol{\ell}} -C)I^{-1}\boldsymbol{e}_k \\ & \qquad\qquad + \beta (\boldsymbol{e}_k^T I^{-1}\boldsymbol{\ell}) \left[\boldsymbol{e}_k^T(\hat{\boldsymbol{\ell}} -C)I^{-1}\boldsymbol{\ell}-E^R_kV\right] \\
=& \boldsymbol{e}_k^TCI^{-1}\boldsymbol{e}_k-\boldsymbol{e}_k^T \hat{\boldsymbol{\ell}}  I^{-1}\boldsymbol{e}_k-\beta \boldsymbol{\ell}^TI^{-T}\boldsymbol{e}_k\boldsymbol{e}_k^TCI^{-1}\boldsymbol{\ell} \\
& \qquad \qquad  +\beta\boldsymbol{\ell}^TI^{-T}\boldsymbol{e}_k \boldsymbol{e}_k^T \hat{\boldsymbol{\ell}}  I^{-1}\boldsymbol{\ell} \\
=& \boldsymbol{e}_k^TCI^{-1}\boldsymbol{e}_k-tr( \hat{\boldsymbol{\ell}}  I^{-1})-\beta \boldsymbol{\ell}^TI^{-T}CI^{-1}\boldsymbol{\ell} \\
=& \boldsymbol{e}_k^TCI^{-1}\boldsymbol{e}_k-\beta \boldsymbol{\ell}^TI^{-T}CI^{-1}\boldsymbol{\ell} 
\end{aligned} 
\end{equation}
The equation $\boldsymbol{e}_k^TI^{-1}\boldsymbol{\ell}=\boldsymbol{\ell}^TI^{-T}\boldsymbol{e}_k$ is used in lines 1-2. 
In line 2, the term $\boldsymbol{e}_k \boldsymbol{e}_k^T$ equals the identity matrix when summed over.
So the cubic term is simplified to $(I^{-1}\boldsymbol{\ell})^T\hat{\boldsymbol{\ell}} I^{-1}\boldsymbol{\ell}=\boldsymbol{\omega}\cdot(\boldsymbol{\ell}\times \boldsymbol{\omega})$ and vanishes.
In lines 3-4, we employ the fact that the trace of a skew-symmetric matrix multiplying a symmetric matrix is 0.

For (\ref{long}) to hold true, the constant term and the quadratic term should vanish for all $\boldsymbol{\ell}$. Substituting (\ref{long_simplify}) into (\ref{long}), the constant term is,
\begin{equation} \label{last1}
    \begin{aligned}
& \boldsymbol{e}_k^TCI^{-1}\boldsymbol{e}_k-\frac{1}{2}\beta\cdot(BB^T)_{kl} \boldsymbol{e}_k^TI^{-1}\boldsymbol{e}_l \\
=&tr\left((C-\frac{1}{2}\beta\cdot BB^T)I^{-1}\right),
\end{aligned}
\end{equation}
and the quadratic term is,
\begin{equation}
    \begin{aligned}
& -\beta \boldsymbol{\ell}^TI^{-T}CI^{-1}\boldsymbol{\ell}+\frac{1}{2}(BB^T)_{kl}\beta^2\boldsymbol{\ell}^T(I^{-T}e_k e_l^TI^{-1})\boldsymbol{\ell} \\
=& \boldsymbol{\ell}^T(-\beta I^{-T}CI^{-1} +\frac{1}{2}(BB^T)_{kl}\beta^2(I^{-T}e_k e_l^TI^{-1}))\boldsymbol{\ell} \\
=& \boldsymbol{\ell}^T(-\beta I^{-T}CI^{-1} +\frac{1}{2}\beta^2(I^{-T}BB^TI^{-1}))\boldsymbol{\ell} \\
=& -\beta\cdot \boldsymbol{\ell}^TI^{-T}(C-\frac{1}{2}\beta \cdot BB^T)I^{-1}\boldsymbol{\ell} \\
=& -\beta\cdot \boldsymbol{\omega}^T(C-\frac{1}{2}\beta \cdot BB^T)\boldsymbol{\omega}.
\end{aligned} \label{last2}
\end{equation}
The only condition for both (\ref{last1}) and (\ref{last2}) to be zero is that $C-\frac{1}{2}\beta \cdot BB^T$ is a skew-symmetric matrix:
\begin{equation}
    (C-\frac{1}{2}\beta \cdot BB^T)=-(C-\frac{1}{2}\beta \cdot BB^T)^T.
\end{equation}
Since $BB^T$ is symmetric, the condition is equivalent to,
\begin{equation}
    C+C^T=\beta \cdot BB^T,
\end{equation}
which ends the proof.
\end{proof}

\section*{APPENDIX III}\label{sec:App3}
In this section, we provide some useful approximation formulas for the Lie derivative of the logarithm of a group element, \textit{i}.\textit{e}., 
\begin{equation}
    \begin{aligned}
   &E^L_i \boldsymbol{x}, \,E^L_i(\boldsymbol{x}\boldsymbol{x}^T),\, E_i^LE_j^L \boldsymbol{x},\, E_i^L E_j^L(\boldsymbol{x}\boldsymbol{x}^T),\\
    &E^R_i \boldsymbol{x}, \,E^R_i(\boldsymbol{x}\boldsymbol{x}^T), \,E_i^RE_j^R \boldsymbol{x},\, E_i^R E_j^R(\boldsymbol{x}\boldsymbol{x}^T).
    \end{aligned}
\end{equation}
where $\boldsymbol{x}(k)\!=\![\log^\vee k]$ is a vector-valued function of $k\!\in \!G$.

The BCH expansion for $Z = \log(e^Xe^Y)$ is,
\begin{multline}
    Z = X + Y + \frac{1}{2}[X,Y] + \frac{1}{12}([X,[X,Y]] + [Y,[Y,X]]) +\\+ \frac{1}{48}([Y,[X,[Y,X]]] + [X,[Y,[Y,X]]]) + \cdots,
\end{multline}
where $[X,Y] = XY - YX$ and the subsequent terms involve at least five appearances of $X$ or $Y$ in the recursive Lie brackets (i.e., for instance $[Y,[Y,[Y,[Y,X]]]]$).

We first see that,
\begin{multline}\label{eq:BCH_left}
    [\log^\vee(e^{-tE_i}\circ k)] - [\log^\vee k] = -t\boldsymbol{e}_i - \frac{t}{2}ad_i\boldsymbol{x} -\\- \frac{t}{12}ad_Xad_X\boldsymbol{e}_i + \mathcal{O}(t,||\boldsymbol{x}||^4),
\end{multline}
and,
\begin{multline}\label{eq:BCH_right}
    [\log^\vee(k \circ e^{tE_i})] - [\log^\vee k] = t\boldsymbol{e}_i - \frac{t}{2}ad_i\boldsymbol{x} +\\+ \frac{t}{12}ad_Xad_X\boldsymbol{e}_i + \mathcal{O}(t,||\boldsymbol{x}||^4).
\end{multline}

From here, we can see that,
\begin{align*}
    E^L_i\boldsymbol{x} &= \lim_{t\rightarrow 0}\left(\frac{[\log^\vee (e^{-tE_i}\circ k)] - [\log^\vee k]}{t}\right) \\&\approx -\boldsymbol{e}_i - \frac{1}{2}ad_i\boldsymbol{x} - \frac{1}{12}ad_Xad_X\boldsymbol{e}_i + \mathcal{O}(||\boldsymbol{x}||^4),
\end{align*}
and,
\begin{align*}
    E^R_i\boldsymbol{x} &= \lim_{t\rightarrow 0}\left(\frac{[\log^\vee (k\circ e^{tE_i})] - [\log^\vee k]}{t}\right) \\&\approx \boldsymbol{e}_i - \frac{1}{2}ad_i\boldsymbol{x} + \frac{1}{12}ad_Xad_X\boldsymbol{e}_i + \mathcal{O}(||\boldsymbol{x}||^4).
\end{align*}

Therefore, with an error of order $\mathcal{O}(||\boldsymbol{x}||^3)$ we obtain,
\begin{align}\nonumber
    E^L_i\left(\boldsymbol{x}\boldsymbol{x}^T\right) &\approx -\boldsymbol{x}\left(\boldsymbol{e}_i + \frac{1}{2}ad_i\boldsymbol{x}\right)^T - \left(\boldsymbol{e}_i + \frac{1}{2}ad_i\boldsymbol{x}\right)\boldsymbol{x}^T\\\label{eq:ELloglog}
    &\approx -\boldsymbol{x}\boldsymbol{e}^T_i - \boldsymbol{e}_i\boldsymbol{x}^T - \frac{1}{2}\boldsymbol{x}\boldsymbol{x}^Tad_i^T - \frac{1}{2}ad_i\boldsymbol{x}\boldsymbol{x}^T.
\end{align}
Similarly we have,
\begin{equation}\label{eq:ERloglog}
    E^R_i(\boldsymbol{x}\boldsymbol{x}^T) = \boldsymbol{x}\boldsymbol{e}^T_i + \boldsymbol{e}_i\boldsymbol{x}^T - \frac{1}{2}\boldsymbol{x}\boldsymbol{x}^Tad_i^T - \frac{1}{2}ad_i\boldsymbol{x}\boldsymbol{x}^T.
\end{equation}

A term that we will appear is, $E^L_iE^L_j\boldsymbol{x}$, which we can expand with an error of $\mathcal{O}(||\boldsymbol{x}||^3)$ to obtain,
\begin{multline}\nonumber
    E^L_iE^L_j\boldsymbol{x} \approx \frac{1}{2}ad_j\left(\boldsymbol{e}_i + \frac{1}{2}ad_i\boldsymbol{x} + \frac{1}{12}ad_Xad_X\boldsymbol{e}_i\right) -\\-  \frac{1}{12}E^L_i(ad_Xad_X)\boldsymbol{e}_j
\end{multline}
Since $ad_X$ is linear in $\boldsymbol{x}$, to ensure that  $E^L_iE^L_j\boldsymbol{x}$ has an error of $\mathcal{O}(||\boldsymbol{x}||^3)$, we can truncate the expansion of $E^L_i(ad_X)$ with an error of $\mathcal{O}(||\boldsymbol{x}||^2)$. Similarly,
\begin{multline}\label{eq:ERiERjlog}
    E^R_iE^R_j\boldsymbol{x} \approx \frac{1}{2}ad_j\left(\boldsymbol{e}_i - \frac{1}{2}ad_i\boldsymbol{x} + \frac{1}{12}ad_Xad_X\boldsymbol{e}_i\right) +\\+ \frac{1}{12}E^R_i(ad_Xad_X)\boldsymbol{e}_j
\end{multline}

The following is true by the product rule:
\begin{multline*}
E^L_iE^L_j(\boldsymbol{x}\boldsymbol{x}^T) = (E^L_iE^L_j\boldsymbol{x})\boldsymbol{x}^T + E^L_j\boldsymbol{x}E^L_i\boldsymbol{x}^T +\\+ E^L_i\boldsymbol{x}E^L_j\boldsymbol{x}^T + \boldsymbol{x}(E^L_iE^L_j\boldsymbol{x}^T),
\end{multline*}
where the $E^L_i\boldsymbol{x}$ term has to be expanded with an error of $\mathcal{O}(||\boldsymbol{x}||^3)$ but $E^L_iE^L_j\boldsymbol{x}$ only needs to have an error of $\mathcal{O}(||\boldsymbol{x}||^2)$ in the expressions in the paper (since it is usually multiplied with another quantity of order $\mathcal{O}(||\boldsymbol{x}||)$). Likewise,
\begin{multline*}
E^R_iE^R_j(\boldsymbol{x}\boldsymbol{x}^T) = (E^R_iE^R_j\boldsymbol{x})\boldsymbol{x}^T + E^R_j\boldsymbol{x}E^R_i\boldsymbol{x}^T +\\+ E^R_i\boldsymbol{x}E^R_j\boldsymbol{x}^T + \boldsymbol{x}(E^R_iE^R_j\boldsymbol{x}^T),
\end{multline*}

The previous two expressions rely on knowing $E^L_iad_X$ and $E^R_iad_X$ with an error of at most $\mathcal{O}(||\boldsymbol{x}||^2)$. We first note that,
\begin{equation}
    ad(c_1X + c_2Y) = c_1ad(X) + c_2ad(Y),
\end{equation}
by linearity for two constants $c_1$ and $c_2$. This can also be seen by assuming an arbitrary $Z$ such that $ad(c_1X + c_2Y)Z = [c_1X+c_2Y,Z] = c_1[X,Z] + c_2[Y,Z] = c_1ad(X)Z + c_2ad(Y)Z = (c_1ad(X) + c_2ad(Y))Z$. Since this holds for any $Z$, the above identity follows. Then,
\begin{align}\nonumber
    E^L_i(ad_X) &= \lim_{t\rightarrow 0}\left(\frac{[ad(\log(e^{-tE_i}\circ k))] - [ad(\log k)]}{t}\right) \\&\approx 
    -ad_i - \frac{1}{2}ad_{[E_i,X]},
\end{align}
and,
\begin{align}\nonumber
    E^R_i(ad_X) &= \lim_{t\rightarrow 0}\left(\frac{[ad(\log(k\circ e^{tE_i}))] - [ad(\log k)]}{t}\right) \\&\approx 
    ad_i - \frac{1}{2}ad_{[E_i,X]}.
\end{align}
Using these results we have, to an error of $\mathcal{O}(||\boldsymbol{x}||^3)$,
\begin{multline}\label{eq:ELiEljlog}
    E^L_iE^L_j\boldsymbol{x} \approx \frac{1}{2}ad_j\boldsymbol{e}_i + \frac{1}{4}ad_jad_i\boldsymbol{x} +\\+ \frac{1}{24}ad_jad_Xad_X\boldsymbol{e}_i + \frac{1}{12}\left(ad_iad_X\boldsymbol{e}_j + ad_Xad_i\boldsymbol{e}_j\right)+\\+ \frac{1}{24}\left(ad_{[E_i,X]}ad_X\boldsymbol{e}_j + ad_Xad_{[E_i,X]}\boldsymbol{e}_j\right),
\end{multline}
and if we define $P$ and $S$ as the following collection of terms:
\begin{multline*}
    P = \frac{1}{2}ad_j\boldsymbol{e}_i\boldsymbol{x}^T + \frac{1}{4}ad_jad_i\boldsymbol{x}\boldsymbol{x}^T +\\+\frac{1}{12}\left(ad_iad_X\boldsymbol{e}_j\boldsymbol{x}^T + ad_Xad_i\boldsymbol{e}_j\boldsymbol{x}^T\right),
\end{multline*}
\begin{multline*}
    S = \boldsymbol{e}_i\boldsymbol{e}_j^T + \frac{1}{2}\left(\boldsymbol{e}_i\boldsymbol{x}^Tad_j^T + ad_i\boldsymbol{x}\boldsymbol{e}_j^T\right) +\\+ \frac{1}{4}ad_i\boldsymbol{x}\boldsymbol{x}^Tad_j^T + \frac{1}{12}\left(\boldsymbol{e}_i\boldsymbol{e}_j^Tad_X^Tad_X^T + ad_Xad_X\boldsymbol{e}_i\boldsymbol{e}^T_j\right),
\end{multline*}
we would have,
\begin{equation}\label{eq:ELiEljloglog}
    E^L_iE^L_j(\boldsymbol{x}\boldsymbol{x}^T) \approx P + P^T + S + S^T.
\end{equation}

\end{document}